\documentclass{article}

\usepackage[english]{babel}

\usepackage[a4paper,top=2cm,bottom=2cm,left=3cm,right=3cm,marginparwidth=1.75cm]{geometry}

\usepackage{amsmath}
\usepackage{amsthm}
\usepackage{latexsym}
\usepackage{graphicx}
\usepackage[colorlinks=true, allcolors=blue]{hyperref}
\usepackage{physics}
\usepackage{amssymb}

\usepackage{mathtools}  
\mathtoolsset{showonlyrefs}

\usepackage[normalem]{ulem}
\usepackage{authblk}

\newcommand{\dif}{\dd} 
\newcommand{\R}{\mathbb{R}}
\newcommand{\N}{\mathbb{N}}
\newcommand{\rem}[1]{}
\newcommand{\der}[2][]{\frac{\dd {#1}}{\dd {#2}}}
\newcommand{\pder}[2][]{\frac{\partial {#1}}{\partial {#2}}}
\DeclareMathOperator{\ad}{\mathrm{ad}}

\newtheorem{remark}{Remark}
\newtheorem{theorem}{Theorem}
\newtheorem{lemma}{Lemma}

\newtheorem{definition}{Definition}

\usepackage{csquotes}
\usepackage[
backend=biber,
style=alphabetic,
]{biblatex}
\title{\textsc{Extremal Black Holes as \\ Relativistic Systems with Kepler Dynamics}}

\author[1,2]{Dijs de Neeling\footnote{d.w.de.neeling@rug.nl}}
\author[1]{Diederik Roest\footnote{d.roest@rug.nl}}
\author[2]{Marcello Seri\footnote{m.seri@rug.nl}}
\author[2]{Holger Waalkens\footnote{h.waalkens@rug.nl}}
\affil[1]{Van Swinderen Institute for Particle Physics and Gravity, \authorcr\small University of Groningen, The Netherlands}
\affil[2]{Bernoulli Institute for Mathematics, Computer Science and Artificial Intelligence, \authorcr\small University of Groningen, The Netherlands}

\begin{document}
\maketitle

 \begin{abstract} \noindent
The recent detection of gravitational waves emanating from inspiralling black hole binaries has triggered a renewed interest in the dynamics of relativistic two-body systems. The conservative part of the latter are given by Hamiltonian systems obtained from so called post-Newtonian expansions of the general relativistic description of black hole binaries. In this paper we study the general question of whether there exist relativistic binaries that display Kepler-like dynamics with elliptical orbits. We show that an orbital equivalence to the Kepler problem indeed exists for relativistic systems with a Hamiltonian of a Kepler-like form. This form is realised by extremal black holes with electric charge and scalar hair to at least first order in the post-Newtonian expansion for arbitrary mass ratios and to all orders in the post-Newtonian expansion in the test-mass limit of the binary. Moreover, to fifth post-Newtonian order, we show that Hamiltonians of the Kepler-like form can be related explicitly through a canonical transformation and time reparametrization to the Kepler problem, and that all Hamiltonians conserving a Laplace-Runge-Lenz-like vector are related in this way to Kepler.\medskip
{\footnotesize\textbf{MSC classes}: 37J06, 70H15, 83C22, 83C57} 
  \end{abstract}


\section{Introduction}

From 2015 onward, the observatories LIGO, VIRGO and KAGRA have detected many instances of gravitational waves originating from binaries of neutron stars or black holes~\cite{LIGOScientific, Abbott_2016}. With more observing runs~\cite{O4LIGO} and the space based telescope LISA upcoming, the dawn of the gravitational wave era provides a strong motivation for precision study of binary dynamics, particularly of the earlier stage of the merger~\cite{2017arXiv170200786A,Amaro-Seoane:2022rxf}.

The early inspiral stage is usually approached with analytical tools~\cite{Isoyama:2020lls}, such as the post-Newtonian (PN) expansion. Though there is a good understanding of the PN expansion and other approximations in the context of binary systems, both of the conservative and radiation part to high order, future experiments that are very sensitive to a long inspiral phase stimulate the further development of analytical tools to limit the demand on computational resources.  

This paper is therefore dedicated to the identification of specific relativistic systems for which the PN expansion results in a system of a much more manageable form. Often, systems with this kind of simplifications possess more symmetries and conservation laws, reducing the number of effective degrees of freedom. This is famously the case in the classical analogue of the relativistic binary systems; we will investigate to what extent the same holds for certain relativistic systems.

\bigskip

On the non-relativistic i.e.~classical level, the two-body problem divides up nicely into the motion of a free particle (the total mass located at the center of mass) and the motion of a particle with reduced mass $\mu=\frac{m_1m_2}{m_1+m_2}$ in the stationary potential generated by the total mass $M=m_1+m_2$. The solutions to this problem then are the same ellipses as in the classical Kepler problem in celestial mechanics. The latter possesses, next to the expected spherical symmetry $SO(3)$ yielding the conservation of angular momentum, an additional symmetry which gives the conservation of the Laplace-Runge-Lenz (LRL) vector. Since this symmetry is not immediately obvious on the level of the Lagrangian it is often referred to as a hidden symmetry. 

The three components of the angular momentum vector, the three components of the LRL vector and the total energy form seven conserved scalar quantities. As the length of the LRL vector is determined by angular momentum and energy, and the angular momentum vector is perpendicular to the LRL vector, only five of the scalar conserved quantities are independent. The joint levelsets of these five constants of motion in the six-dimensional phase space are hence one-dimensional. As a consequence, bounded orbits must be periodic and take the form of the famous elliptical orbits found in Kepler's model of the Solar System (while in General Relativity of course, the symmetry is broken and the perihelion -- the point of closest approach -- precesses).

For central force systems, there is a strong link between closing orbits and enhanced symmetry, in the form of Bertrand's theorem. This states that the only two central forces whose bounded orbits are all closed curves are the Kepler potential and the isotropic harmonic oscillator~\cite{Goldstein2001-ql}, which are in fact related (see e.g.~\cite{VANDERMEER2015181}) and known for their large symmetry groups, $SO(4)$ and $SU(3)$, respectively. 

Since we know symmetries make problems more tractable and the non-relativistic problem possesses additional symmetry, it is natural to attempt to restore the non-relativistic symmetry in relativistic systems. While the closing of bounded orbits is not a sufficient condition for conservation of a LRL vector, it is a necessary condition that is satisfied very rarely by relativistic theories. The closure of orbits can therefore be a useful tool for diagnosis of theories when looking for additional spacetime symmetries, as demonstrated by e.g.~\cite{Caron-Huot:2018ape}. 

There has been previous work done on identifying relativistic systems that have exclusively closed bounded orbits. For example, Perlick has identified all spacetimes in General Relativity with that property, a sort of relativistic Bertrand theorem~\cite{Perlick_1992}. He considered all spherically symmetric spacetimes that have bounded timelike geodesics with a perihelion  shift equal to $\frac{\pi}{\beta}$, with $\beta$ rational. The cases $\beta=1$ and $\beta=2$, corresponding to relativistic versions of the Kepler problem and harmonic oscillator, are the only ones admitting an additional symmetry. However, Perlick's theorem is only taking into account gravity, without allowing other forces to be present. Additionally, there is the hydrogen-like system in $\mathcal{N}=4$ super Yang-Mills theory, which has an additional conserved vector, coinciding with the classical LRL vector in the non-relativistic limit~\cite{Caron-Huot:2014gia,Alvarez-Jimenez:2018lff}. Interesting follow-up results were derived in $\mathcal{N}=8$ supergravity, where the two-body problem was shown to have a LRL vector to at least order $1$PN and a vanishing perihelion shift to third order in the post-Minkowskian (PM) expansion\footnote{The post-Newtonian expansion is in terms of $\frac{1}{c^2}$, resulting in an expansion in weak gravitational field and low velocity, while the post-Minkowskian is an expansion in gravitational constant $G$, i.e. a weak gravitational field expansion only~\cite{Damour:2016gwp}.}~\cite{Caron-Huot:2018ape,Parra-Martinez:2020dzs}. However, at $3$PM there appears to be a hint that the quantum energy level degeneracy linked to the LRL vector and present at $1$PN might be lost. This suggests an interesting break in the bond between closed bounded orbits and hidden symmetry, which is present classically. Additionally, it was shown in~\cite{Caron-Huot:2018ape} that the test-mass limit in $\mathcal{N}=8$ supergravity has a zero perihelion shift to all orders in velocity.

\bigskip

Although relativistic corrections of the Kepler problem generically break the symmetry associated with the LRL vector, it follows from the above that specific systems manage to preserve it in some sense. These systems then, one might wonder, are perhaps not truly relativistic in some sense, as their dynamics is still constrained by the same symmetries, giving rise to strictly periodic orbits in phase space (at least for bounded orbits). 

We will study a class of such systems and demonstrate that they are orbital equivalent to the Kepler system on a levelset of the Hamiltonian in phase space. Their full Hamiltonians are implicitly defined by
\begin{align} \label{eq: IH}
  f(H(q,p)) = \frac{p^2}{2} -  \frac{g(H(q,p))}{r(q)} \,,
\end{align}
for smooth functions $f,g:\R \to \R$. As we will see, such systems give rise to a phase space which can be thought of as being foliated by energy surfaces of Kepler problems where for each value of $H$ the motion is proportional to that of a Kepler problem with a different coefficient for the gravitational potential; in other words, with a different gravitational constant. The global structure of the phase space is therefore tied to the specific properties of the function $g(H)$.

We will show that the above class of Hamiltonians~\eqref{eq: IH} naturally arises in Einstein-Maxwell-dilaton (EMD) theory, when one considers two extremal black holes with opposite charge for a specific value of the dilaton coupling (cf. equation \eqref{eq: specEMD2}
). For this case, we derive a functional relation of the form \eqref{eq: IH}  to first order in the post-Newtonian expansion of the two-body system and all orders in the test-mass limit. 

The case of extremal black holes is somewhat special and unlikely to be realized in nature (with all observed black holes approximately neutral), and the dilaton coupling of $a=\sqrt{3}$ (see Section~\ref{Two-body}) is much larger than current experimental bounds~\cite{Casini}. However, as a physical aside, it is an interesting question how one would observationally distinguish the above Hamiltonians from the Kepler one, e.g.~in the solar system. When studying planets orbiting a (much heavier) Sun described by $H(q,p)$ as opposed to the ordinary Kepler problem, the first two laws of Kepler still hold: the bounded orbits are ellipses and the trajectories conserve angular momentum. However, the period of an orbit becomes
\begin{equation}
    T=2\pi\sqrt{\frac{s^3}{GMg(E)}},
\end{equation}
with $s$ the semi-major axis of the ellipse, where for the sake of clarity we included the mass of the sun $M$ and the gravitational constant $G$. This differs from the usual Newtonian period $T_N=2\pi\sqrt{{s^3} / {(GM)}}$. Therefore the third law\footnote{For all bodies orbiting the Sun, the square of the period is proportional to the third power of the semi-major axis of the orbit, {\it with the same proportionality constant}~\cite{Kepler}.} is violated: different orbits will have different energies, causing the ratio $T^2/s^3$ to no longer be the same for all orbiting objects. 

To what extent, then, are these Hamiltonians equivalent to the Kepler problem? We will prove they describe the same dynamics at least on the energy shell, so for a fixed $H=E$, in the sense that their flows are proportional. Moreover, there can exist a transformation mapping $H$ to the Kepler Hamiltonian where we do not need to restrict to the energy surface (at least locally, in a small neighbourhood of the energy surface). This transformation is shown to exist as an energy redefinition and canonical transformation at least up to and including 5PN. 
It is worth noting that our local constructions show the existence of Kepler dynamics and the conserved LRL vector, but do not necessarily imply global $SO(4)$ symmetry. Proving the transformations generated by the conserved quantities form a group of canonical transformations isomorphic to $SO(4)$ is not trivial~\cite{Bacry}. Our construction only shows the local existence of the so(4) Lie algebra, generated by the angular momentum and the appropriately rescaled LRL vector. Whether the approximate transformations to Kepler can be extended to all orders, whether they exists globally and whether the symmetry group is indeed $SO(4)$, remain topics of future research.

Another question is whether relativistic systems canonically conjugate to Kepler up to time reparametrization are \emph{the only} ones with a conserved LRL-type vector and the corresponding symmetry. We show this is the case at least to 5PN order. 

\bigskip

This paper is organised as follows. In Section~\ref{Kepler-type hams}, we discuss the set of Kepler-like Hamiltonians and their relation to the Kepler problem. Here we show the on-shell equivalence to Kepler in Subsection~\ref{on-shell}, a construction that yields an explicit off-shell transformation to the Kepler problem in Subsection~\ref{off-shell} and the equivalence of all LRL-preserving Hamiltonians of a certain kind to Kepler in Subsection~\ref{HidSym}. After a physical intermezzo in Section~\ref{EMD sec} introducing Einstein-Maxwell-dilaton black holes, we show in Section~\ref{Two-body} that a particular tuning of the parameters in this theory allows one to write the 1PN two body and all-order test-mass limit as a Kepler-type Hamiltonian, providing an interesting example of relativistic systems with classical dynamics.

\bigskip
\textbf{Notational conventions}. In what follows, unless differently specified, we will use the following notational conventions. 

We will use $(q,p)$ to denote canonical coordinates in $T^\star\R^3\cong \R^3\times\R^3$. The radial momentum is denoted $p_r$, i.e. $ p_r= \frac{(p\cdot q) }{r}$  where $r = r(q) = \abs{q}$. We will use upper indices to denote vector components, therefore, $V^i$ denotes the $i$th component of a vector $V$. Indices for relativistic objects are denoted by $\mu$ and $\nu$. Throughout, we will assume Einstein notation and omit explicit sums. For instance, 4-vectors are denoted by $x^\mu$, the Lorentzian metric is denoted $g_{\mu\nu}$ and, therefore, the inner product of tangent vectors with respect to $g$ is given by $g_{\mu\nu}\dot{x}^\mu\dot{y}^\nu$. For the metric we assume the signature $(- + + + )$.

For convenience, throughout the paper we use units such that the speed of light, $c$, and the gravitational constant, $G$, are equal to 1.


\section{Relativistic Systems with Kepler Dynamics} \label{Kepler-type hams}
In this section we will discuss Hamiltonians of type~\eqref{eq: IH} central to this work. Although they appear naturally from relativistic problems, see Section~\ref{Two-body}, they end up being equivalent (in the ways mentioned in the introduction) to the classical Kepler system. We will first review `relativistic' corrections within the post-Newtonian (PN) expansion that we will employ. Subsequently we will prove the equivalence on the energy surface (on-shell) of the Kepler-like Hamiltonians to Kepler problems and, later, how these can be explicitly related (up to fifth PN order) to the Kepler problem through canonical transformations and a non-linear energy redefinition (off-shell). Lastly, we show that all Hamiltonians preserving a LRL-like vector are related to Kepler in this way, also up to 5PN.


\subsection{The post-Newtonian expansion}\label{sec:PN}
The post-Newtonian expansion involves a power-series expansion in the small parameter $\frac{1}{c^2}$, which physically amounts to an asymptotic expansion both in slow motion and weak field. Thus the $0$PN order accounts for the non-relativistic terms and the highest-order in an $n$PN expansion comprises terms proportional to $\left(\frac{1}{c^2}\right)^n$. This ordering has been in use for binary systems for a long time, going back to Einstein in calculating the anomalous precession of Mercury~\cite{Einstein}. 

A general two-body Hamiltonian has translational and rotation symmetry. The reduction of the translational symmetry can be accomplished by choosing a center-of-mass frame. In a center-of-mass frame a general two-body Hamiltonian can then be written solely in terms of the $SO(3)$ invariants $p^2$, $q^2$ and $(p\cdot q)$. Knowing the PN orders of these terms individually would then give us a way to count the orders of all terms in a general two-body Hamiltonian. We can infer these orders by inspection of some simple cases. A relativistic free particle, for example, has Hamiltonian
\begin{align}
  H_{\text{fp}}=mc^2\sqrt{1+\frac{p^2}{m^2c^2}}
  =m c^2 \left( 1 + \frac{p^2}{2m^2 c^2} - \frac{p^4}{8m^4c^4} + \mathcal{O} \big( \frac{p^6}{m^6c^6}  \big) \right)\,.
\end{align}
Note that each momentum appears only in the dimensionless combination $\frac{p^2}{m^2c^2}$, and that the leading term in this expansion gives the rest-mass energy $mc^2$. When including the gravitational pull of a large mass $M$ on a test body $m$, similar considerations show that the radius $r=\abs{q}$ only appears in the combination $\frac{2GM}{r c^2}$. The third type of term that can be present in relativistic systems with spherical symmetry is the inner product $(p\cdot q)$, appearing only in the radial momentum $p_r=\frac{(p\cdot q)}{r}$, which carries a $\frac{1}{c}$ to be consistent with the total momentum. These considerations give us an order counting system for terms occurring in a relativistic Hamiltonian. Such a Hamiltonian is given by
\begin{equation}\label{eq: PN-counting}
    H_{\rm rel}=\frac{1}{\epsilon}\left(1+ \sum_{j=0}^\infty \epsilon^{j+1}\Lambda_j(\alpha)\right),
    \qquad\Lambda_j(\alpha)= \sum_{\substack{(l,m,n) \in\mathbb{N}^{3} \\  l+m+n=j+1}} \alpha_{l,m,n}\frac{(p^2)^l(p_r^2)^n}{r^m}\,,
\end{equation}
where we introduced $\epsilon=\frac{1}{c^2}$ as a bookkeeping parameter to easily keep track of the PN orders. We set $m=1$. Because of the rest-mass term, clearly the PN orders of $\Lambda_j(\alpha)$ will be shifted down by one. As the constant mass term does not influence the dynamics, we will drop it from now on in our mathematical analysis, but remember the effect the overall factor  $\frac{1}{\epsilon}$ in \eqref{eq: PN-counting} when we discuss PN orders. We note that we will not discuss any systems where spin plays a role.

To mathematically make sense of this, we will give a formal definition of post-Newtonian expansion of functions on phase space.
\begin{definition} \label{def: PN}
A Hamiltonian function $B(\epsilon; q,p)$ depending on a small parameter $\epsilon > 0$ is in \textbf{post-Newtonian expansion to $N$th order} if it is of the form 
\begin{equation}\label{eq: PNdef}
    B(\epsilon; q,p)=\sum_{j=0}^N \epsilon^{j}B_j(q,p) + \mathcal{O}(\epsilon^{N+1})\,,
\end{equation} 
for some regular enough Hamiltonian functions $B_j$.
\end{definition}

The term \emph{post-Newtonian} comes from the fact that the expansions is carried out in the context of General Relativity where $\epsilon=1/c^2$
and the speed of light $c$ is considered large with $\epsilon=0$ yielding the classical limit of General Relativity. 


\subsection{On-shell equivalence to Kepler dynamics} \label{on-shell}
Let us consider the following family of Hamiltonians $H = H_{f,g}: T^*\R_0^d\simeq (\R^d-\{0\})\times\R^d \to \R$, implicitly defined by the functional relation
\begin{align}\label{eq: implicit}
  f(H(q,p)) = \frac{p^2}{2} -  \frac{g(H(q,p))}{r(q)} \,,
\end{align}
with $f,g:\R\to\R$ smooth functions, that can be written as powers series in the form  
\begin{equation} \label{eq: pow ser}
    f(x) =x+f_1 x^2+ f_2 x^3 +\dots\,, \qquad 
    g(x) = 1+g_1 x+g_2 x^2+\dots\,,  
\end{equation}
where $f_i, g_i$ are real numbers. In other words, we assume that their Taylor-Maclaurin expansions have their first coefficients fixed by $f_{-1}=0$, $g_0=1$ and $f_0=1$ (with labels related to PN orders as will become more clear below). The reason the constant term in $f(H(q,p))$ with coefficient $f_{-1}$ is absent, is that we want to disregard rest-mass terms and~\eqref{eq: implicit} to yield the Kepler system at lowest order in the PN expansion below.

As discussed in the introduction, this is directly motivated by the Hamiltonian of a binary Einstein-Maxwell-dilaton system in the test-mass limit. In Section~\ref{Two-body}, we will see this system has a Hamiltonian that is implicitly defined (at all PN orders) by the above relation with 
\begin{align}
 f(x) = x + \tfrac{1}{2}x^2, \qquad g(x) = 1+x+\tfrac{1}{4}x^2\,.
\end{align}
Note that we set the test mass to unity in the identifications, to exactly match the description of the above Hamiltonian family. 

Since we already know explicit Hamiltonian functions solving~\eqref{eq: implicit}, we will not pursue the question of sufficient conditions for existence\footnote{This is a rather interesting technical problem on its own, and we refer the interested reader to~\cite{cristea2017,BergerAller2021} for the current state of the art.}. For the time being, we assume that a solution $H(q,p)$ exists for the given functions $f, g$ and describe some of its properties in relation with the Kepler Hamiltonian. Therefore, let $H: T^*\R_0^d \to \R$ be a smooth Hamiltonian function satisfying the relation~\eqref{eq: implicit}. For convenience, we define the new Hamiltonian function
\begin{equation}\label{eq: K}
    K(q,p):=f(H(q,p))=\frac{p^2}{2}-\frac{g(H(q,p))}{r(q)}\,.
\end{equation}
Since $K$ is by definition a function of $H$, it is also an integral of motion of $H$, that is, $K$ is constant on the flow of $H$. If $E\in\R$ is a regular value of $H$, this implies that on the energy levels $H^{-1}(E)$
\begin{equation}
  K|_{H^{-1}(E)}(q,p) = \frac{p^2}{2}-\frac{g(E)}{r(q)}\,,
\end{equation}
which is a Kepler-type Hamiltonian with gravitational constant $g(E)$. In fact, the flow generated by $K$ on all its regular energy surfaces turns out to be proportional to the flow of a Kepler Hamiltonian. 

\begin{theorem}\label{theorem:FunctionalRelation}
Consider $M=T^*\R_0^d \simeq (\R^d-\{0\})\times\R^d$ with standard symplectic form $\omega=\sum_k \dif p_k \wedge \dif q_k$. Assume that there are functions $f,g\in C^2(\R)$ such that the Hamiltonian $H:M\to \R$ is defined by \eqref{eq: implicit} and 
\[
K(q,p) := f(H(q,p))\,.
\]
Then, for any regular energy value $H=E$ the vector fields $X_K|_{H^{-1}(E)}$ and $X_H|_{H^{-1}(E)}$ of $K$ and $H$ are proportional.
Moreover, if
\begin{equation}\label{def:regular levelset}
  \mathcal{E} := \left\{
    E\in\R \mid f'(E)\neq 0 \mbox{ and }E\mbox{ is regular value of } H
  \right\}
\end{equation}
and
\[
  J: M\times\mathcal{E}\to\R, \qquad
  J(q,p,E):= J_E(q,p) := \frac{p^2}{2} - \frac{g(E)}{r(q)}\,,
\]
then for all $E \in \mathcal{E}$, the Hamiltonian vector fields $X_{J_E}|_{H^{-1}(E)}$ and $X_K|_{H^{-1}(E)}$ are proportional.
\end{theorem}
\begin{proof}
For the Hamiltonian vector fields $X_K$ and $X_H$ of $K$ and $H$, respectively, we have
\begin{equation}
\label{eq:proportional}
X_K = -\pder[K]{q} \pder{p} +\pder[K]{p} \pder{q} = f'(H) \left(-\pder[H]{q} \pder{p}+\pder[H]{p}\pder{q}\right) = f'(H)  X_H\,.
\end{equation}
The vector fields are hence proportional.

For the second part, let $E\in\mathcal{E}$. The Hamiltonian vector field of $J_E$ at a point $(q,p)$ is
\[
  X_{J_E}(q,p) = - \pder[J(q,p,E)]{q} \pder{p} + \pder[J(q,p,E)]{p} \pder{q}\,.
\]
Using $K(H(q,p))=J(q,p,H(q,p))$ we have
\[
\begin{split}
  X_K &=  -\left( \pder[J]{q} +  \pder[J]{H} \pder[H]{q} \right)\pder{p} + \left( \pder[J]{p}+ \pder[J]{H} \pder[H]{p}\right) \pder{q}  \\ 
     &= \left(- \pder[J]{q} \pder{p} +  \pder[J]{p} \pder{q} \right) +  \pder[J]{H} \left( - \pder[H]{q} \pder{p}  + \pder[H]{p} \pder{q} \right)\,.
\end{split}
\]
Therefore
\[
  X_K|_{H^{-1}(E)} = X_{J_E}|_{H^{-1}(E)} + \frac{ \pder[J]{E}|_{H^{-1}(E)} }{f'(E)} X_K|_{H^{-1}(E)} \,,
\]
where we used~\eqref{eq:proportional} for the second term. 
Solving for $X_{J_E}$
we get that on $H^{-1}(E)$ 
\begin{equation}\label{eq:X_K-X_J_E-proportionality}
  \left(1- \frac{ \pder[J]{E}}{f'(H)}  \right) X_K = X_{J_E}\,.
\end{equation}
\end{proof}
\noindent
This means that the evolution of Hamiltonians satisfying~\eqref{eq: implicit} is equivalent to the evolution of a classical Kepler problem - more precisely, for each energy, the trajectories are equivalent to that of a Kepler problem with a specific energy-dependent value of the coupling with the potential up to possibly a time-rescaling. In particular, all bounded orbits are ellipses in the configuration space. This relation is somewhat reminiscent of the Maupertuis-Jacobi transformation, in which the trajectories of a natural Hamiltonian are described via a time reparametrization as geodesics of a metric~\cite{Tsiganov_2001,Chanda:2016aph}.

The fact that the above Hamiltonians are equivalent to the Kepler problem and, in particular, the fact that all trajectories are closed, hints at the existence of an associated conserved Laplace-Runge-Lenz (LRL) vector on the energy levels. An obvious candidate would be the vector 
\begin{equation}\label{eq: GenLRL}
  A^i({q},{p})=({p}\cross L)^i({q},{p})-g(H({q},{p}))\frac{{q^i}}{r({q})} \,,
\end{equation}
since it is simply the classical LRL vector, with an additional coefficient corresponding to the coefficient of the potential energy in~\eqref{eq: K}.
\begin{theorem}
Let $\mathcal{E}$ be defined by~\eqref{def:regular levelset} and $E \in \mathcal{E}$.
On the set
\[
    \left\{
    (q,p)\in H^{-1}(E) \;\big|\; f'(E)-\pder[J(q,p,E)]{E} \neq 0
    \right\},
\]
the Hamiltonian $H(q,p)$ defined in~\eqref{eq: implicit} is in involution with all components of the Laplace-Runge-Lenz vector $A^i(q,p)$ defined in~\eqref{eq: GenLRL}, and hence these are integrals of motion of the dynamics generated by $H$. 
\end{theorem}
\begin{proof}
  Fix $E\in\mathcal{E}$. On $H^{-1}(E)$ the flow of $J_E(q,p)$ is proportional to that of $H(q,p)$, with proportionality 
\begin{equation}
  \lambda := \left(1-\frac{\frac{\partial J(q,p,E)}{\partial E}}{f'(E)}\right)\,,
  \end{equation}
which we assume to be regular and nonvanishing. 
we know for the Lie derivatives of the functions $A^i(q,p)$ with respect to $X_H$
\begin{equation}\label{eq:HA->JA}
  \{ A^i,H\} = \mathcal{L}_{X_H}(A^i)
    = \mathcal{L}_{\lambda^{-1}X_{J_E}}(A^i)
    = \lambda^{-1}\mathcal{L}_{X_{J_E}}(A^i)
    = \lambda^{-1}\{A^i,J_E\}\,,
\end{equation}
where all functions are evaluated on $H^{-1}(E)$.

With~\eqref{eq:HA->JA}, we reduced ourselves to check whether $J_E(q,p)$ commutes with the components of the LRL vector on $H^{-1}(E)$.
Namely,
\begin{equation}\label{eq:proofThm1}
\begin{aligned}
      \{A^i, J_E(q,p)\}&=\pder[A^i]{q}\pder[J_E]{p} - \pder[A^i]{p} \pder[J_E]{q}\\
                   &=\left[\pder{q}(p\cross{L})^i -\left(\pder{q}\frac{q^i}{r}\right)g(H)\right]\pder[J_E]{p} - \left[\pder{p} (p\cross{L})^i-\left(\pder{p}\frac{q^i}{r}\right)g(H)\right]\pder[J_E]{q} \\
                    &\quad + \left(-\frac{q^i}{r}\right)\left[\pder[g(H)]{q} \pder[J_E]{p} -\pder[g(H)]{p} \pder[J_E]{q}\right]\,.                           
\end{aligned}
\end{equation}
Observe now that on $H^{-1}(E)$, $A^i = A_E^i := (p\cross {L})^i-g(E)\frac{q^i}{r}$. So the first two terms combine into the Poisson bracket
\begin{equation}
  \left\{A_E^i(q,p), J_E(q,p)\right\}
  = \left\{(p\cross{L})^i-g(E)\frac{q^i}{r},\frac{p^2}{2}-g(E)\frac{1}{r}\right\}
\end{equation}
which vanishes as the Poisson bracket between a standard Kepler Hamiltonian and its LRL vector.

The square bracket that form the last term in \eqref{eq:proofThm1}  amounts to $\{g(H), J_E\}$ evaluated on ${H^{-1}(E)}$. This also vanishes due to $g(H)=g(f^{-1}(K))$ and application of Theorem~\ref{theorem:FunctionalRelation}.
\end{proof}

\begin{theorem}
On regular levelsets of $H$ such that $f(H)<0$, the rescaled LRL vector defined by
$\Bar{A}^i :=-\frac{A^i}{\sqrt{-2f(H)}}$, where $A$ is defined in \eqref{eq: GenLRL},
satisfies the following commutation relations
\begin{equation}\label{eq:LRL Lie Alg}
    \{L^i,L^j\}=\epsilon_{ijk}L^k,\qquad \{\Bar{A}^i,\Bar{A}^j\}=\epsilon_{ijk}L^k,\qquad \{L^i,\Bar{A}^j\}=\epsilon_{ijk}\Bar{A}^k\,,
\end{equation}
which define a Lie algebra isomorphic to so(4).
\end{theorem}

\begin{remark}
To study the global symmetry, one would need to regularize the problem to complete the temporal flow and then consider the global transformations generated by the integrals above~\cite{Ligon1976}, which is out of the scope of this paper.
\end{remark}

\begin{proof}
Let $E$ be a regular value for $H$ and such that $f(E)<0$, in the rest of this proof we assume all the computations restricted to the levelset $H^{-1}(E)$.
For the first relation, the calculation is the same as in the usual Kepler case. Writing out the second, we have 
\begin{equation}
\begin{split}
    \{\Bar{A}^i,\Bar{A}^j\}&=\frac{1}{-2f(H)}\{A^i_E,A_E^j\} +\frac{1}{-2f(H)}\left(\frac{f'}{-2f(H)}A^i-\frac{q^i}{r}g'\right)\{H,A_E^j\}\\
    &+\frac{1}{-2f(H)}\left(\frac{f'}{-2f(H)}A^j-\frac{q^j}{r}g'\right)\{A_E^i,H\}\\
    &+\frac{1}{-2f(H)}\left(\frac{f'}{-2f(H)}A^i-\frac{q^i}{r}g'\right)\left(\frac{f'}{-2f(H)}A^j-\frac{q^j}{r}g'\right)\{H,H\}\,,
\end{split}
\end{equation}
where $f':=f'(H)$ and $g':=g'(H)$. The second and third term here vanish due to proportionality of $H$ to $J_E$, which commutes with $A_E^i$, while the last term vanishes trivially. Since we are on a fixed levelset, we can consider $g(E)$ constant in the first term, and therefore the bracket yields, as for the usual Kepler problem,
\begin{equation}
    \{A^i_E,A_E^j\}=-2\left(\frac{p^2}{2}-\frac{g(E)}{r}\right)\epsilon_{ijk}L^k = -2 f(E)\, \epsilon_{ijk}L^k \,,
\end{equation}
completing the calculation of the bracket of rescaled LRL vectors. The remaining bracket again reduces to the computation for the classical Kepler system in complete analogy to the above computation.

The commutation relations~\eqref{eq:LRL Lie Alg} then define a Lie algebra isomorphic to so(4) as shown e.g. in~\cite[Chapter 3.2, Proof (3.6)]
{Cushman1997}.
\end{proof}

While in this section we proved that for each fixed value of the energy the family of Hamiltonians satisfying~\eqref{eq: implicit} has a flow which is proportional to the Kepler flow and admits a LRL vector, we do not know the regularity of the dependence of these objects on the energy itself nor how to relate~\eqref{eq: implicit} and a Kepler Hamiltonian beyond the energy surface. In the following section, we will consider this problem, looking for an energy-independent way to relate Kepler problems and the implicitly defined Hamiltonians~\eqref{eq: implicit}.

What we can immediately observe is that, while the shape of orbits is the same in both~\eqref{eq: implicit} and in a Kepler problem, the energy levelsets $H_{\rm Kep}^{-1}(E)$ and $H^{-1}(E)$ foliate the phase space in a different way. The Hamiltonian $K(q,p)$, and therefore also the implicitly defined Hamiltonians~\eqref{eq: implicit}, induce a bundle of non-equivalent Kepler orbits, the global structure of which is determined by $g(H)$.


\subsection{Off-shell equivalence to Kepler dynamics} \label{off-shell}

While the on-shell equivalence discussed in the previous subsection explains why the Hamiltonians implicitly defined by~\eqref{eq: implicit} have an additional constant of motion and hence closed orbits, it does not address the violation of Kepler's third law: the fact that Keplerian energy surfaces can be stacked differently in the Kepler bundle. We now turn to this issue, and address the question whether one can also map families of orbits with different energies onto a fixed Kepler system. Since we would like to avoid issues of singularities and/or topology, we will restrict ourselves to a local construction. In other words, we now aim to generalise the on-shell (on a fixed energy surface) orbital equivalence  to an off-shell equivalence (for a neighbourhood of orbits of possibly different energies).

The violation of Kepler's third law demonstrates a physical difference between the Kepler problem and the implicitly defined Hamiltonians on the phase space, so it should not come as a surprise that looking for such a relation will involve a transformation of the phase space itself. The mapping we are looking for therefore involves both a time reparametrization (related to the mapping from $H$ to $K \equiv f(H)$) as well as a canonical transformation, whose composition will (locally) transform the Kepler Hamiltonian to the implicitly defined ones and establish an orbital equivalence in this sense.

We will provide evidence for the existence of such a canonical transformation by explicitly constructing it up to fifth PN order. Note that the PN expansion differs from the expansion around an energy surface; even when extending the canonical transformation to all PN orders (or having a closed expression for it), this would still only involve a local equivalence as singularities or topological issues might prevent one to extend the mapping to the whole phase space.  
Addressing the extension to all PN orders and the question of convergence of the series constructed below (even just in an asymptotic sense) is not a trivial endeavour, as is the question of global existence of the phase space transformation. Therefore, we will leave the all-order analysis for the whole phase space for future research. 

\bigskip
 
The goal is to find a solution $H$ to the functional equation \eqref{eq: implicit} to any desired PN order from the perturbation of the Kepler system. To find a perturbative solution for $H$, it is useful to take $H$ itself dimensionless\footnote{We divide out the rest-mass energy $mc^2$ and set $m=1$ as previously.}, but explicitly include the PN expansion parameter $\epsilon$, see Section~\ref{sec:PN}:
\begin{equation} \label{eq: pert}
    f(H)=\epsilon \frac{p^2}{2}-g(H)\frac{\epsilon}{r}\,,
\end{equation}
which is solved to lowest order by
\begin{equation}\label{eq: KepHam}
H=\epsilon H_{\rm Kep}= \epsilon \left(\frac{p^2}{2}-\frac{1}{r}\right)\,,
\end{equation}
that is, the Kepler Hamiltonian with an extra factor $\epsilon$. More specifically we will show the following.
 
  \begin{theorem} \label{theorem:Kepler_equivalence}
For given $C^\infty$ functions $f$ and $g$, the relation \eqref{eq: pert} can be solved to at least PN order 5 by 
\begin{equation}\label{eq:H_from_H_Kep}
H=\Phi^\star \tau(\epsilon H_{\text Kep})\,,
\end{equation}
where $\tau:\R\to\R, \,E \mapsto\tau(E)$ is $C^\infty$ with $\tau'(0)=1$ defining a near-identity time re-parametrization  and
$\Phi$ is a near-identity canonical transformation. 
\end{theorem}
 \begin{remark}
As we work with dimensionless Hamiltonian $H$, the lowest order term in $H$ is order $\epsilon$. Therefore, PN order 5 here corresponds to order $\epsilon^6$. When going back to a dimensionfull Hamiltonian, this additional factor of $\epsilon$ vanishes. This understanding of PN orders coincides with the `relative $k$PN order' of~\cite{Tanay}.
 \end{remark}
 
The proof will be given by an explicit construction based on Lie transform perturbation theory combined with a rescaling of the energy function. 

Let us introduce the real vector spaces 
 \begin{equation}\label{def:Wj}
     W_j = \text{span} \left\{  \frac{(p^2)^l(p\cdot q)^n}{r^m} \;\big|\; (l,m,n)\in \N^3, \, l+m-\frac12n=j+1\right\}\,.
 \end{equation}
 For instance, the Kepler Hamiltonian $H_{\rm Kep}$ is in $W_0$. We will mainly consider $W_j$ with non-negative integer $j$ resulting from even  $n$ in \eqref{def:Wj} such that $F_j\in W_j$ has PN order $j$ (see the discussion following \eqref{eq: PN-counting}). But as we will see below, also half-integer $j$ resulting from odd $n$ in \eqref{def:Wj} can be important. 
Note that for $F_i\in W_i$ and $F_j\in W_j$, 
\begin{equation}\label{eq:FiFj}
\{F_i,F_j\}\in W_{i+j+\frac32}\,,    
\end{equation}
implying that 
$$
W=\bigoplus_{k\in \N}W_{k/2}
$$ 
is closed under the Poisson bracket. In particular, 
$$
\{p\cdot q,F_j\} \in W_{j}\,.
$$
Let us write the energy rescaling $\tau$ in \eqref{eq:H_from_H_Kep} in a power series as 
\begin{equation} \label{eq: EnRedef}
   \tau(E) =\sum_{n=0}^\infty \delta_n  E^{n+1} \,,
\end{equation}
with $\delta_0=1$. For counting the PN orders of $\tau$ applied to some Hamiltonian function $H$ it is important to note that for $F_i\in W_i$ and $F_j\in W_j$,
$$
F_i\,F_j \in W_{i+j+1}\,,
$$
which implies that $W$ is also closed under multiplication.

We will consider a succession of near-identity  canonical transformations each of which is obtained from the flow of the Hamiltonian vector field generated by a suitable function $G$. To describe the construction it is useful to introduce the adjoint operator 
\begin{equation}
    \ad_G(\cdot):=\{G,\cdot\}\,.
\end{equation}
For convenience, we define $n>0$ repeated iterations of the adjoint operator by 
\begin{equation}
[\ad_G]^n := [\ad_G]^{n-1} \circ \ad_G, \qquad
[\ad_{G}]^0 := \mathrm{Id}_{C^\infty(M)}\,.
\end{equation}

Under the canonical transformation given by the time-one map of the flow generated by the Hamiltonian $G$ a function $F$ transforms according to
\begin{equation}\label{eq:GeneralCanTrafo}
    F\mapsto \Phi^\star F=\sum_{m=0}^\infty \frac{
    1}{m!}[\ad_G]^m F.
\end{equation}
From \eqref{eq:FiFj} we get that for $F_i\in W_i$ and $F_j\in W_j$,
\begin{equation}\label{eq:adFiFj}
    \ad_{F_i}(F_j)\in W_{i+j+\frac32}
\end{equation}
The idea now is to solve the functional relation \eqref{eq: implicit} 
order by order with $H_\text{Kep}$
through a succession of canonical transformations generated by functions $G_i$ and an energy rescaling of the form \eqref{eq: EnRedef}.
To this end let us first inspect the functional relation in terms of the power series for $f$ and $g$ in \eqref{eq: pow ser} which gives
 \begin{equation} \label{eq: impl exp}
    H+f_1 H^2+ f_2 H^3 +\dots= \epsilon\frac{p^2}{2}-\left(1+g_1 H+g_2 H^2+\dots\right)\frac{\epsilon}{r}\,.
\end{equation}
In order to find solutions for integer PN orders we will need terms $\ad_{G_i}(H_{\rm Kep})$ in the canonical transformations to yield integer order and hence the $G_i$ to be have half-integer order (see \eqref{eq:adFiFj}). It turns out that this can be achieved by the ansatz
\begin{equation} \label{eq: GenGen}
    G_{i-\frac12}(q,p)=(p\cdot q) \Lambda_i(a)(q,p)\,,
\end{equation}
where $\Lambda_i(a)$ again denotes a general function of order $\epsilon^i$ with coefficients $a_{l,m,n}$ as  
defined in \eqref{eq: PN-counting}.

Each such $G_{i-\frac12}$ generates a canonical transformation $\Phi_i$ and will be determined such that
\begin{equation} \label{eq:H_sought}
H:=\Phi_n^\star\cdots\Phi_2^\star\Phi_1^\star \tau(\epsilon H_\text{Kep})\,,
\end{equation}
with suitable $\delta_i$ in \eqref{eq: EnRedef} defining the energy rescaling $\tau$ solves the functional relation \eqref{eq: impl exp} to order $n$.

\begin{lemma} \label{lemma:oders_trans_H}
For positive integers $k$ and $i_1\le i_2\le\ldots\le i_m$, let $I=(i_m,\ldots,i_2,i_1)$
\begin{equation}
    \ad^I_{G_I} := \ad_{G_{i_m-\frac12}}\dots\ad_{G_{i_2-\frac12}}\ad_{G_{i_1-\frac12}}
\end{equation} 
with $G_{i_k-\frac12}\in W_{i_k-\frac12}$, $k=1,\ldots,m$, and $\abs{I}=\sum_{k=1}^m i_k$. 
    Let $H$ be defined as in \eqref{eq:H_sought}. Then the PN expansion of $H$ to order $N$ is given by 
\begin{equation} \label{eq: RedAndTrans}
\sum_{j=0}^N \epsilon^{j+1} H_j\, ,
\end{equation}
where
\begin{equation}
    H_j := \sum_{n=0}^{j} \sum_{k=0}^{j-n} \sideset{}{'}\sum_{\substack{I \in\mathbb{N}_+^{k} \\  \abs{I}=j-n}} \frac{\delta_n}{k!} \ad_G^I\left(H_{\rm Kep}^{n+1}\right)
\end{equation}
is in $W_j$. Here the prime in the third sum denotes that the summation is restricted to tuples of ordered integers
$I=(i_k,\ldots,i_2,i_1)\in\mathbb{N}_+^{k}$ with $i_1 \le i_2\le\ldots\le i_k$.
\end{lemma}
\begin{proof}
    The result follows immediately from ordering the terms in \eqref{eq:H_sought} taking into account  \eqref{eq:GeneralCanTrafo} and \eqref{eq:adFiFj}.
\end{proof}

We now come to the proof of Theorem~\ref{theorem:Kepler_equivalence}.

\begin{proof}
The proof is done by explicit computation. 

From Lemma~\ref{lemma:oders_trans_H} we get
\begin{equation}  \label{eq:H_oder_epsilon2}
\begin{split}
H =\quad &\phantom{+} \epsilon H_\text{Kep} \\ 
&+\epsilon^2 \big(\{G_{1-\frac12},H_\text{Kep}\} + \delta_1 H^2_\text{Kep} \big)\\
&+\epsilon^3\big(\delta_2 H_{\rm Kep}^3 +\{G_{2-\frac12},H_{\rm Kep}\}  
        +\{G_{1-\frac12},\delta_1 H_{\rm Kep}^2\}+\frac{1}{2}\{G_{1-\frac12},\{G_{1-\frac12},H_{\rm Kep}\}\} \big) \\
&+O(\epsilon^4)\,.
\end{split}
\end{equation}
A fast way to proceed is to rewrite the functional relation \eqref{eq: impl exp} as
\begin{equation}\label{eq:iteration_proced}
    H= \epsilon\frac{p^2}{2}-\left(1+g_1 H+g_2 H^2+\dots\right)\frac{\epsilon}{r} -f_1 H^2- f_2 H^3 -\dots
\end{equation}
Equating the right hand sides of \eqref{eq:H_oder_epsilon2} and \eqref{eq:iteration_proced} at order $\epsilon$ gives 
$$
 H_0=\left(\frac{p^2}{2}-\frac{1}{r}\right)\,.
$$
Plugging this $H_0$ into the right hand side of
\eqref{eq:iteration_proced}, reading off the terms 
of order $\epsilon^2$ and equating with the order $\epsilon^2$ in \eqref{eq:H_oder_epsilon2}
gives
\begin{equation}
     -f_1 H_{\rm Kep}^2-g_1 H_{\rm Kep}\frac{1}{r}=\delta_1 H_{\rm Kep}^2+ \{G_{1-\frac12},H_{\rm Kep}\} \,.
\end{equation}
This is solved
by choosing the coefficient of the energy redefinition as 
 \begin{equation}
      \delta_1 = 2 g_1-f_1  
\end{equation}
and the generating function
\begin{equation} \label{eq: G_1}
    G_{1-\frac12}= - g_1 (p\cdot q) H_{\rm Kep}\,. 
\end{equation}
Filling in the $\epsilon H_0 + \epsilon^2 H_1$  into the right hand side of
\eqref{eq:iteration_proced}, reading off the terms 
of order $\epsilon^3$ and equating with the order $\epsilon^3$ in \eqref{eq:H_oder_epsilon2}
gives
\begin{equation}
\begin{split}
        (2f_1^2-f_2)H_{\rm Kep}^3+(3f_1g_1-g_2)\frac{H_{\rm Kep}^2}{r}+g_1^2\frac{H_{\rm Kep}}{r^2}&= \delta_2 H_{\rm Kep}^3 +\{G_{2-\frac12},H_{\rm Kep}\} \\ 
        &+\{G_{1-\frac12},\delta_1 H_{\rm Kep}^2\}+\frac{1}{2}\{G_{1-\frac12},\{G_{1-\frac12},H_{\rm Kep}\}\} \,.
\end{split}
\end{equation}
This can be solved by choosing the next coefficient in the energy rescaling as
\begin{equation}
 \begin{split}
      \delta_2 = 5 g_1^2-6 g_1 f_1+2 g_2+2 f_1^2-f_2
 \end{split}
 \end{equation}
 and the generating function
 \begin{equation}
 \begin{split}
            G_{2-\frac12} = (p\cdot q) \left(\frac{1}{2} \left(-g_1^2+2 g_1 f_1-2 g_2\right) H_{\rm Kep}^2+\frac{\frac{g_1^2}{2} H_{\rm Kep}}{r}\right) \,.
 \end{split}
 \end{equation}
We have carried out the computation to order 5PN ($\epsilon^6$) with the help of \textsf{ Mathematica} and present the computations and results in ancillary files~\cite{TODO}.
\end{proof}
We note that, assuming the particular form~\eqref{eq: GenLRL} of the LRL vector, one can show this is conserved off-shell as well, to at least order 5PN.


\subsection{Hidden symmetries require Kepler dynamics} \label{HidSym}

In the previous sections we have described a class of relativistic Hamiltonians that turns out to be equivalent to a classical Kepler problem, either being proportional to it on its energy levels or using an approximate canonical transformation and time reparametrization. In both cases, we also constructed the modified LRL vector.

In this section, we aim to investigate a more general question: what is the largest class of relativistic two-body Hamiltonians (within a certain set of plausible Hamiltonians) that share the symmetries of the Kepler problem? And secondly, is this class related to the Kepler system through canonical transformation and time reparametrization? This would in effect generalise an aspect of Bertrand's Theorem, as all Hamiltonians obeying the symmetries would be equivalent to the Kepler problem - just like in the classical context the only system obeying the symmetries (which requires a vanishing perihelion shift) is Kepler.

\begin{theorem} \label{theorem: LRL}
Let a spherically symmetric class of relativistic two-body Hamiltonians be given by
\begin{equation}
    H=\frac{1}{\epsilon}\left(\epsilon H_0+\epsilon^2H_1+\epsilon^3H_2+\dots\right)\,,
\end{equation}
where
\begin{equation} \label{eq: GenHam}
    H_0=\frac{p^2}{2}-\frac{1}{r}
    \qquad
    \mbox{and, for $j\geq1$,}
    \qquad
    H_j=\Lambda_{j}(c)\,.
\end{equation}
At least up to and including 5PN, these Hamiltonians are canonically conjugate to Kepler up to time reparametrization if and only if they conserve an extra vector, which to leading order is given by
\begin{equation}
    A^i_0(q,p)=(p\cross L)^i-\frac{q^i}{r}\,,
\end{equation}
and may contain (and in general will contain) corrections at higher orders.    
\end{theorem}
This additional conserved vector can be seen as a relativistic version of the Laplace-Runge-Lenz vector.

\begin{remark} \label{CoefConstr}
There are two notions we can use to constrain the number of coefficients $c_{l,m,n}$ in the Hamiltonian~\eqref{eq: GenHam} (see also~\eqref{eq: PN-counting}). Firstly, if a particle is far away from any gravitational body, one expects the attraction to become negligible and the Hamiltonian to approach the special relativistic Hamiltonian, being an expansion only in $p^2$. Therefore, all momentum-only terms must be solely dependent on the regular momentum and independent of the radial momentum $p_r$. This removes all terms with coefficients $c_{l,0,n}$ where $n\neq 0$. Secondly, the first order Hamiltonian is known to possess an ambiguity, allowing one to shift the radius such that the term $\sim p_r^2/r$ vanishes~\cite{HiidaOkamura,Bjerrum-Bohr:2002gqz}. This kind of ambiguity can be expected also at higher orders, but finding these is not a trivial task and not necessary for our analysis.
\end{remark}
\begin{proof}
While the second part of Theorem~\ref{theorem: LRL} of course is a tautology - all systems conjugate to Kepler problems conserve the symmetries of Kepler problems -, the first part is not at all obvious. We have checked the statement up to fifth post-Newtonian order, and will show here the first two.

To prove the equivalence, we want to show that the Hamiltonian of a candidate symmetric system can be related through a time reparametrization and a canonical transformation to the Kepler system. Since the Hamiltonians are divided up in separate orders, we can demand this transformation exists on each order individually. We find the symmetric Hamiltonian and its associated LRL vector by taking an Ansatz for the vector and its corrections and requiring they commute up to an order. That gives a set of relations among both the coefficients $c_{l,m,n}$ of the Hamiltonian~\eqref{eq: GenHam} and $\alpha_{l,m,n}$, $ \beta_{l,m,n}$ of the LRL Ansatz, of which the $\epsilon^{j}$-order term is given by
\begin{equation} \label{eq: LRLAns}
    A_j^i= \Lambda_{j-1}(\alpha)(p\cdot q) p^i+\Lambda_{j}(\beta) q^i\,,
\end{equation} 
where we assume~\eqref{eq: PN-counting}. The Hamiltonian given in terms of the remaining free variables can then be matched to the time-reparametrized, canonically transformed Kepler Hamiltonian, constructed in the same way as in the previous subsection. 

At first non-leading order for example, the terms proportional to $\epsilon$ take the form
\begin{equation}
    \{H,A^i\}=\{ H_0,\epsilon A^i_1\}+\{\epsilon H_1,A^i_0\} = 0\,,
\end{equation}
which results in relations among the $3$ coefficients $c_{l,m,n}$ of the Hamiltonian (see Remark~\ref{CoefConstr}) and the $9$ coefficients of the Ansatz for the LRL vector at first order. The existence of a LRL vector up to this order requires it to take the following form
\begin{equation}
\begin{aligned}
\beta _{1,1,0}&= 3 \alpha _{1,0,0}+\beta _{1,0,0} c_{1,1,0}+4 \beta _{1,0,0} c_{2,0,0}
\,,
&\beta _{2,0,0}&= -\alpha _{1,0,0}\,,\\
\beta _{0,2,0}&= -2 \left(\alpha _{1,0,0}+\beta _{1,0,0} c_{1,1,0}+4 \beta _{1,0,0} c_{2,0,0}\right)\,,
&\alpha _{0,1,0}&= -2 \alpha _{1,0,0}
\,,
\end{aligned}
\end{equation}
with all other coefficients vanishing because of the choices in Remark~\ref{CoefConstr}. As $\beta_{1,0,0}$ is the parameter determining the overall size of the vector, the only free parameter that is left over is $\alpha_{1,0,0}$. The term in the LRL corresponding to this parameter turns out to be proportional to $A^i_0 H_0$. These are trivially commuting quantities with $H_0$, so we can set the coefficient to $0$, yielding an expression completely fixed in terms of two of the coefficients of the Hamiltonian. The corrected vector is then a conserved quantity only provided we constrain the Hamiltonian with
\begin{equation}
c_{0,2,0}=-2 \left(c_{1,1,0}+2 c_{2,0,0}\right)\,.
\end{equation}
Using the general generating function~\eqref{eq: GenGen} we can then obtain the transformations needed to produce the above Hamiltonian from the Kepler Hamiltonian. These transformations are defined by 
\begin{equation}
\begin{aligned}
    \delta_1&= -4 \left(c_{1,1,0}+3 c_{2,0,0}\right)\,, &a_{1,0,0}&= c_{1,1,0}+4 c_{2,0,0}\,,\\ 
    a_{0,1,0}&= -2 \left(c_{1,1,0}+4 c_{2,0,0}\right)\,, &a_{0,0,1}&= 0\,.
\end{aligned}
\end{equation}

Following the same procedure at second order, the equation that needs to be satisfied is
\begin{equation}
       \{H,A^i\}=\{ H_0, \epsilon^2A^i_2\}+\{ \epsilon H_1, \epsilon A^i_1\}+\{ \epsilon^2H_2,A^i_0\} = 0\,.
\end{equation}
Now there are $7$ coefficients for the Hamiltonian and $16$ for the LRL. This leads to $15$ constraints on the coefficients of the LRL vector, given in the ancillary files~\cite{TODO}. Once more, the remaining degree of freedom is proportional to a vector trivially commuting with $H_0$, that is $A^i_0H_0^2$, and we can set the corresponding coefficient to $0$. The transformation then yields a conserved quantity provided we impose the two constraints:
\begin{equation}
\begin{aligned}
c_{0,2,1}&= -2 c_{1,1,0}^2-16 c_{2,0,0} c_{1,1,0}-32 c_{2,0,0}^2-3 c_{0,3,0}-c_{1,1,1}-5 c_{1,2,0}-8 c_{2,1,0}-12 c_{3,0,0}\,,\\
c_{0,1,2}&= 2 c_{1,1,0}^2+16 c_{2,0,0} c_{1,1,0}+32 c_{2,0,0}^2+c_{0,3,0}-c_{1,1,1}+c_{1,2,0}-4 c_{3,0,0}\,,
\end{aligned}
\end{equation}
and thus five free parameters at this order remain in the Hamiltonian. To relate this to Kepler, we need the transformations given by
\begin{equation}
    \begin{aligned}
    \delta_2&= -4 \left(-c_{1,1,0}^2-8 c_{2,0,0} c_{1,1,0}-16 c_{2,0,0}^2+2 c_{0,3,0}+2 c_{1,2,0}+2 c_{2,1,0}+2 c_{3,0,0}\right)\,,\\
    a_{2,0,0}&= \frac{1}{2} \left(3 c_{1,1,0}^2+16 c_{2,0,0} c_{1,1,0}+16 c_{2,0,0}^2+2 c_{0,3,0}+2 c_{1,2,0}+2 c_{2,1,0}+4 c_{3,0,0}\right)\,,\\
    a_{1,1,0}&= -7 c_{1,1,0}^2-40 c_{2,0,0} c_{1,1,0}-48 c_{2,0,0}^2-5 c_{0,3,0}-5 c_{1,2,0}-4 c_{2,1,0}-4 c_{3,0,0}\,,\\
    a_{0,2,0}&= 8 c_{1,1,0}^2+48 c_{2,0,0} c_{1,1,0}+64 c_{2,0,0}^2+7 c_{0,3,0}+8 c_{1,2,0}+8 c_{2,1,0}+8 c_{3,0,0}\,,\\
    a_{0,1,1}&= \frac{1}{3} \left(-2 c_{1,1,0}^2-16 c_{2,0,0} c_{1,1,0}-32 c_{2,0,0}^2-c_{0,3,0}+c_{1,1,1}-c_{1,2,0}+4 c_{3,0,0}\right)\,.\\
    \end{aligned}
\end{equation}

Similar calculations confirm up to and including fifth PN order that general Hamiltonians of the type described above conserving a LRL vector are related to Kepler via a canonical transformation and time reparametrization. The ancillary files \cite{TODO} include a \textsf{Mathematica} notebook with the higher order computations and the resulting Hamiltonians and transformations.    
\end{proof}
In other words, this theorem indicates that there are no free lunches: only systems that are canonically conjugate to Kepler up to time reparametrization have the same extension of the spatial rotation {algebra so(3) with hidden symmetries to so(4)}.

\section{Dilaton-Coupled Einstein-Maxwell Theory} \label{EMD sec}

This section is included as a physics-oriented intermezzo, presenting some background on the relevant physical system to be discussed in Section~\ref{Two-body}. It may be skipped without much harm to the understanding of our main conclusions.

We will focus on the Einstein-Maxwell-dilaton (EMD) theory; a generalisation of general relativity that is of interest due to its general nature of forces, comprising spin-0,1,2 background fields, as well as its ability to circumvent the ``no-hair theorem". This states that a black hole cannot be described by properties other than its mass, charge, and angular momentum. EMD escapes this prohibition by introducing a non-trivial scalar field and charge\footnote{This scalar field gives the black hole what is called \emph{secondary} hair, as the scalar charge is completely determined in terms of mass and charge within a given theory, i.e. for a given value of the scalar coupling constant~\cite{Coleman:1991ku}.}. As we will see, these features will lead to interesting aspects in terms of black hole orbits. In this section, we will describe the different black hole solutions as well as their source terms.

\subsection{Black holes with dilaton hair}

The fields present in EMD theory are the metric $g_{\mu\nu}$, the four-potential $A_\mu$ and the dilaton field $\phi$, exponentially coupled (through coupling constant $a$) to the electromagnetic field strength. Mathematically, these are respectively a pseudo-metric tensor, a covector and a function on the space-time manifold $\R^4$. A solution of the corresponding Einstein-Maxwell's equation is then given by the critical tensors of the action (see also e.g.~\cite{Holzhey, Cornish})
\begin{equation}
  S[g_{\mu\nu},A_\mu,\phi]=\frac{1}{16\pi}\int \ \dif^4 x \sqrt{-g}\left(R-2(\partial\phi)^2-e^{-2a\phi}F^2\right) \,,
\end{equation}
where $g = \det(g_{\mu\nu})$ is the determinant of the pesudometric-tensor matrix, $R$ is the Ricci scalar and $F = dA$ is Maxwell's field.
In addition to diffeomorphism invariance and gauge symmetry, this action has a global symmetry that shifts the dilaton while rescaling the gauge vector.
In what follows we will omit further mathematical details and focus on a more physical description.

Starting with the special case $a=0$, the static and spherically symmetric solutions of this theory are given by the well-known Schwarzschild and Reissner-Nordstr\"{o}m black holes, with possibly non-vanishing electric charge. In the electrically neutral case, the introduction of the dilaton does not introduce additional solutions; scalar-gravity is known to satisfy the no-hair theorem and hence cannot carry scalar charge~\cite{Chrusciel:2012jk}. In contrast, when introducing the dilaton (i.e.~$a \neq 0$) in the charged case, the solution becomes more interesting and reads~\cite{Maeda, Horowitz}
\begin{align}
    \dif s^2 & = -\lambda^2\dif t^2+\lambda^{-2}\dif r^2+ r^2\kappa^2 \dif \Omega^2 \,, \qquad
    F_{tr}=\frac{e^{2a\phi_0}Q}{r^2\kappa^2} \,, \qquad
    e^{2a\phi} =e^{2a\phi_0}\left(1-\frac{r_-}{r}\right)^{\frac{2a^2}{1+a^2}} \,,
\end{align}
where 
\begin{align}
\kappa^2 = \left(1-\frac{r_-}{r}\right)^{\frac{2a^2}{1+a^2}} \,, \qquad
    \lambda^2 =\left(1-\frac{r_+}{r}\right)\left(1-\frac{r_-}{r}\right)^{\frac{1-a^2}{1+a^2}} \,.
\end{align}
Note that one can set $\phi_0=0$ by the shift symmetry of the dilaton, which we will subsequently do. This most general solution is parametrised by the locations of the inner and outer horizons $r_\pm$. These are related to the mass and charge of the object by
\begin{align} 
    r_+ & =M+\sqrt{M^2+Q^2(a^2-1)} \,, \qquad
    r_- =\left(\frac{a^2+1}{a^2-1}\right)\left(-M+\sqrt{M^2+Q^2(a^2-1)}\right)\,, \label{eq:mass}
\end{align}
Importantly, the `horizons' labelled by the minus sign are singular for all $a>0$ (i.e., the scalar curvature diverges at this point), whereas the ones labelled by the plus signs are not.

As mentioned above, this solution carries scalar charge, given by a simple integration over a spherical shell surrounding it~\cite{Horowitz}:
\begin{align}
    D&=\lim_{\rho\to\infty}\frac{1}{4\pi}\oint \ \nabla^\mu \phi\ \dif^2 \sigma_\mu 
        =\frac{a}{{a^2-1}} \left(-M + \sqrt{M^2 + Q^2 (a^2-1)}\right)\,. \label{dilaton-charge}
\end{align}
In order to have non-vanishing dilaton charge, one therefore needs both electric charge $Q \neq 0$ as well as non-vanishing scalar coupling $a \neq 0$. For a given theory and hence value of $a$, the mass and charge determine the dilaton charge, which is therefore not an independent parameter.

For completeness, we would like to mention that for the same set of charges ($M,Q,D)$, a second solution exists, given by the above fields but with parameters
\begin{align}
    \tilde{r}_+&=M-\sqrt{M^2+Q^2(a^2-1)}\,, \qquad
    \tilde{r}_-=\left(\frac{a^2+1}{a^2-1}\right)\left(-M-\sqrt{M^2+Q^2(a^2-1)}\right) \,,
\end{align}
where we have added a tilde to avoid confusion with the solutions that form our main interest. In the neutral case, these solutions are the Janis-Newman-Winicour solution for Einstein minimally coupled to a scalar field~\cite{Luna:2016hge}. Note that they are in general different from Schwarzschild (when choosing $a \neq 0$); however, in this case the solution develops a naked singularity (that is, a non-removable singularity not cloaked by an event horizon). This amounts to the statement that scalar-gravity does not have any black hole solutions other than Schwarzschild. The introduction of the electric charge does not qualitatively change this singular property. For these reasons we will not consider this solution any further.


\subsection{The extremal case}

We now turn to the extremal case of the hairy black hole solutions~\eqref{eq:mass}. To this end, it is convenient to rewrite the relation~\eqref{dilaton-charge} between the three charges as the quadratic relation
 \begin{align}
     (D - a M)^2 = a^2 (M^2 + D^2 - Q^2)  \,. \label{charge-relation}
    \end{align}
The importance of the expression on the right-hand side  lies in the extremality of the black hole. Imagine two such black holes; when this combination vanishes, attractive spin-0,2 forces between two such black holes (proportional to $M^2 + D^2$) would exactly cancel the repulsive spin-1 force (proportional to $Q^2$). 

The dimensionless parameter
 \begin{align}
  \chi^2 \equiv \frac{M^2 + D^2 - Q^2}{M^2} \,,
 \end{align}
is therefore a measure of extremality, and interpolates between $0$ and $1$. When $\chi = 1$, this corresponds to a neutral black hole (i.e.~the Schwarzschild solution). In contrast, the case $\chi = 0$ corresponds to an extremal black hole: in this case, the two sides of~\eqref{charge-relation} vanish separately, and the black hole has extremal charges
 \begin{align}
    D_{\rm extr} = a M \,, \qquad Q_{\rm extr} = \pm \sqrt{1+a^2} M \,, \label{extremal-charge}
 \end{align}
that are both linearly proportional to the mass. For all values $a\neq 0$, the solutions will be singular in the extremal limit \cite{Ortin}. Moreover, the thermodynamics of such extremal objects are fundamentally different for $a\gtrless 1$ -- in fact, it has been argued~\cite{Holzhey} that they resemble elementary particles more than black holes for $a>1$. As this will be of no consequence for the dynamics, which is our concern here, we will still refer to these objects as black holes.

Due the cancellation of forces between extremal black holes, one can also construct multi-center solutions. For the Einstein-Maxwell case, these are the Majumdar-Papapetrou solutions \cite{Maj47,Pap45}, while the solutions with a non-minimally coupled dilaton field added in have been discussed in~\cite{Cornish}. The line element in this case is given by
\begin{equation} \label{eq: CGmetric}
    \dif s^2 = - U^{-2/(1+a^2)}\dif t^2 +U^{2/(1+a^2)} \dif q\cdot\dif q\,,
\end{equation}
with 
\begin{equation}
   U(q) = 1+(1+a^2)\sum_{n}\frac{M_n}{\abs{q-q_n}}\,,
\end{equation}
where the sum is over the extremal black holes with mass $M_n$ and positions $q_n$ of which there may be arbitrarily many. The no-force condition implies that all centers carry electric and dilaton charges~\eqref{extremal-charge} that are proportional to their masses. For a single charge, this solution corresponds to the extremal case of the general Einstein-Maxwell-dilaton metric. This can be seen by noting the two horizons of the EMD merge into $r_\pm=M(1+a^2)$ and switching to the isotropic radius $\rho=r(1-\frac{r_\pm}{r})$.


\subsection{Skeletonisation} \label{skeletonisation}

To make our dynamical systems pertain to dilaton-charged black holes, simply taking a point particle with a mass and electric charge while keeping the universal dilaton coupling $a$ nonzero does not suffice. For self gravitating objects, even in the zero size limit, there will be a dependence on the background scalar field of the way the object couples to it. One can see this by considering the black hole presented in the first subsection: both the dilaton charge and electric charge depend on the background scalar field, while the electric charge is conserved by a $U(1)$ symmetry.  

In general then, one can describe a particle by its conserved charge $Q_p$ and a mass function $\mathfrak{m}(\phi)$, absorbing the dependence on the dilaton field. It will show up in the Lagrangian describing the dynamics of the point particle, reading
\begin{equation} \label{eq: PP}
    L_{pp}=\mathfrak{m}(\phi)\sqrt{-g_{\mu\nu}\dot{x}^\mu\dot{x}^\nu}+Q_pA_\mu\dot{x}^\mu.
\end{equation}
We can compare the field generated by a particle in this parametrisation to the field we know belongs to a certain object, in order to find the mass function belonging to the zero size limit of the particular object. Taking a black hole as example, this leads to the matching condition~\cite{Julie, Khalil2018}
\begin{equation} \label{eq: matching}
    \frac{\dif \mathfrak{m}(\phi)}{\dif \phi} = \frac{a}{{a^2-1}} \left(- \mathfrak{m}(\phi) + \sqrt{\mathfrak{m}(\phi)^2 + Q_p^2 e^{2 a \phi} (a^2-1)}\right) \,.
\end{equation}
For every value of the dilaton coupling $a$, the solution to this equation will depend on the charge and an integration constant, determined by the mass $m$ and charge. 

Note that the above ODE is fully analogous to the expression for the dilaton charge~\eqref{dilaton-charge}, with the identifications  
 \begin{align}
   (M,D,Q) \simeq (\mathfrak{m}(\phi) , \frac{d \mathfrak{m}(\phi)}{d \phi} , Q_p e^{a\phi} ) \,.
 \end{align}
Indeed, one should think of the latter as the background-dependent charges, which go to their asymptotic values for $\phi \rightarrow 0$. The mass function therefore determines more than only the masses. Its first derivative corresponds to the dilaton charge. Moreover, its second derivative is closely related to the extremality combination:
 \begin{align}
   \frac{d^2}{d \phi^2} \log \mathfrak{m}(\phi) =  \frac{a^2 Q_p^2 e^{2a\phi}}{\mathfrak{m}^2(\phi)} \frac{(a-\frac{\dif }{\dif \phi})\mathfrak{m}(\phi)}{a\mathfrak{m}(\phi)+(a^2-1)\frac{\dif \mathfrak{m}(\phi)}{\dif \phi}}  \,,
 \end{align}
which is evaluated on the background to be
\begin{equation}
    \beta:=\frac{d^2}{d \phi^2} \log \mathfrak{m}(\phi)\vert_{\phi=0}=\frac{a^2Q^2}{m^2}\frac{\chi}{\chi+\frac{aD}{m}}\,,
\end{equation}
clearly vanishing for extremal black holes. 

In the extremal case, therefore, the coupling of the particle to the dilaton field is simply through an exponential $\mathfrak{m}(\phi)=m e^{a\phi}$. Looking at the Lagrangian~\eqref{eq: PP}, this shows the extremal particle couples like a particle without self-gravitation to the metric and dilaton field. In retrospect this is not surprising, since, if the extremal particle does not experience a net force from other extremal particles stationary with respect to it, why would it experience any force generated by itself? Equivalently, the extremal particle can be seen to couple to a metric given by
\begin{equation}
    \tilde{g}_{\mu\nu}=e^{2a\phi}g_{\mu\nu}\,.
\end{equation}
If we make the transformation to the tilde metric, we switch from the Einstein frame to the Jordan frame, in which the bulk action takes the form~\cite{Flanagan2004,Khalil2018}
\begin{equation}
    S_{\text{Jordan}}=\int \ \dif^4 \mathbf{x} \sqrt{-\tilde{g}}\ e^{-2a\phi}\left(\tilde{R}+\left(6a^2-2\right)(\partial\phi)^2-F^2\right)\,.
\end{equation}
In this frame, the extremal particle does not couple to the dilaton field at all, making its mass constant 
\begin{remark}
In general, for non-extremal cases,~\eqref{eq: matching} has no simple closed-form expression. An exception is the case $a=1$, for which it is solved by $\mathfrak{m}(\phi)^2 = \mu^2+\frac{1}{2}Q_p^2e^{2\phi}$, where the integration constant $\mu$ is given by $\mu^2=m^2-\frac{Q_p^2}{2}$, showing it is a measure of deviation from extremality, since for $a=1$ the particle is extremal when $Q_p^2=2m^2$ (setting the background field to zero). 
\end{remark}

\section{The Two-Body System of Extremal Black Holes} \label{Two-body}
Following the physics intermezzo, we now return to the main theme of this paper -- the analysis and understanding of Hamiltonians with Kepler-like dynamics -- and employ the dynamics of black holes in Einstein-Maxwell-dilaton gravity as an example. We will consider a pair of non-spinning black holes in EMD theory, carrying both electric and dilatonic charge besides their mass. 

In the first part of this section, we restrict ourselves to the first order in the post-Newtonian expansion, i.e.~at 1PN. As we will see, for a specific case of the dilaton coupling $a$ and extremal charges, this system coincides with a Kepler-like system. In the second part, we show how the same equivalence to Kepler dynamics arises in a different region in parameter space: instead of 1PN for arbitrary mass ratio, we now focus on the test-mass limit with a vanishing mass ratio, or $m_1\ll m_2$. This corresponds to the motion of a charged particle in a given background as outlined in Section~\ref{EMD sec}, and can be studied at all orders in the post-Newtonian expansion. Prompted by the two-body discussion, we will focus specifically on extremal black holes with opposite charges.


\subsection{Kepler dynamics at 1PN}

For the two-body system with arbitrary masses $m_{1,2}$, electric charges $Q_{1,2}$ and dilaton charges $D_{1,2}$ (subject to the relation~\eqref{dilaton-charge}), the 0PN Hamiltonian in center-of-mass coordinates reads
 \begin{align}
        H_{0PN} = & \frac{p^2}{2\mu}-\frac{G_{12}M\mu}{r} \,,
 \end{align}
where the effective Newton's constant is given by the interplay between attractive and repulsive forces,
\begin{align}
    G_{12}&=\frac{1}{m_1 m_2} (m_1 m_2 + D_1 D_2 - Q_1 Q_2) \,.
\end{align}
Moreover, we introduce the total mass, reduced mass and symmetric mass ratio given by
\begin{equation}
    M=m_1+m_2\,,\qquad \mu=\frac{m_1m_2}{M}\,,\qquad \nu=\frac{\mu}{M} \,,
\end{equation}
in the usual way.

The 1PN Hamiltonian can be found in e.g.~\cite{Khalil2018} and can be written in terms of three terms\footnote{Note that this in general will also have an additional $p_r^2 /r$ term, proportional to the radial momentum only. By means of a constant shift of the radial coordinate, one can set the coefficient of this term to zero, see e.g.~\cite{Bjerrum-Bohr:2002gqz}. We will do so in order to facilitate the comparison to Section~\ref{Kepler-type hams}.}
\begin{equation}
    H_{1\text{PN}}=h_1\frac{p^4}{4\mu^3}+h_2\frac{\gamma}{\mu^2}\frac{p^2}{r}+h_3\frac{\gamma^2}{\mu r^2}\,,
\end{equation}
writing $\gamma=G_{12}M\mu$ and with dimensionless coefficients given by 
\begin{equation}
    \begin{aligned}
            h_1&=-\frac{1}{2}(1-3\nu) \,, \quad h_2=-\frac{1}{2}\left(\frac{3-D_1D_2}{G_{12}}\right)-\nu\,,\\ h_3&=\frac{\nu}{2}+\frac{1}{2G^2_{12}}
            \left[(1+D_1D_2)^2-2Q_1Q_2
        +\left\{\frac{m_1}{M}(D_1^2\beta_2 +Q_1^2(1+aD_2)-2Q_1Q_2 aD_1)
        +(1\leftrightarrow 2)\right\}\right]\,.
    \end{aligned}
\end{equation}
All quantities here are asymptotic values, as measured far away from any dilaton charge. Moreover, note that we introduce a slight abuse in notation in the above and hereafter to switch to charges and dilaton charges per unit mass, as in $\tilde{Q}_{1,2}=Q_{1,2}/m_{1,2}$ but dropped the tilde to avoid cumbersome expressions.

A comparison to the 1PN Kepler-type Hamiltonians discussed in Section~\ref{Kepler-type hams} demonstrates that these have two free parameters at every order (including 1PN), while the two-body system here has three terms. For general values, this system will therefore not be related to Kepler via a symplectic transformation. More precisely, the linear combination\footnote{This corresponds to the combination $A+2B+C+D$ in the conventions of~\cite{nabet2014leading, Caron-Huot:2018ape}.}
\begin{align}
        \Delta = & h_1+2h_2+h_3 \,, \notag \\
         = & -\frac{1}{2 G_{12}^2}\bigg(6(1-Q_1Q_2)+Q_1^2Q_2^2+2D_1D_2(2-D_1D_2) \notag\\
     & +\left\{\frac{m_1}{M}\left(-D_1^2\beta_2-Q_1^2(1+aD_2)+2Q_1Q_2aD_1\bigg)+(1 \leftrightarrow 2)\right\} \right)\,, 
\end{align}
quantifies the deviation away from Kepler-like dynamics:
 \begin{itemize}
     \item 
When $\Delta$ vanishes, the Hamiltonian can be written in the form~\eqref{eq: implicit} (up to 1PN order), identifying
\begin{equation}
    f_1=-  h_1,\qquad g_1=-2(h_1+h_2)\,.
\end{equation}
In order to see this explicitly, one needs to set $\mu=1$ and scale the quantity $GM$ with $G$ Newton's constant to $\frac{1}{8}$\footnote{This is because the effective Newton's constant $G_{12}$ is eight times larger than the usual gravitational constant. This corresponds to the findings of~\cite{Caron-Huot:2018ape}, who also found this in their supergravity system.}. Hence there exists a canonical transformation to Kepler and the system has a LRL vector. The form of both the canonical transformation and the conserved charge follow from the discussion in Section~\ref{Kepler-type hams}. 
 \item 
In contrast, when $\Delta$ is non-vanishing, the relativistic corrections of this system are not of the Kepler-like form and the corresponding dynamical system differs from Kepler. 
\end{itemize}
The same quantity also determines whether or not bound states have closed orbits: in general they will not, with a perihelion precession given by\footnote{This result has been derived before in~\cite{Kan:2011aau}, though with a different mass function in the sense of Section~\ref{skeletonisation}, such that the results only coincide for extremal black holes.}
\begin{equation} \label{eq: 2bperi}
\begin{aligned}
    \delta \phi_{\rm 1PN, EMD} &= -\frac{2\pi\gamma^2}{ L^2} \Delta \,,
\end{aligned}
\end{equation}
as also stressed by~\cite{Caron-Huot:2018ape}. As a consistency check, let us point out that the GR limit, where all parameters except $m_1, m_2$ and $L$ vanish, reduces to 
\begin{equation}
    \delta \phi_{\rm 1PN, GR} = 6\pi\frac{M^2\mu^2}{L^2}\,,
\end{equation}
as already found by Einstein. Also, the perihelion in Einstein-Maxwell theory, so with dilaton vanishing, becomes
\begin{equation}
    \delta \phi_{\rm 1PN, EM} =\pi\frac{M^2\mu^2}{L^2}\left(6(1-Q_1Q_2)+Q_1^2Q_2^2-\frac{\left(m_1Q_1^2+m_2Q_2^2\right)}{M}\right)\,,
\end{equation}
which in the limit that one mass is much larger than the other agrees with~\cite{Balakin_2000}.

At this point it might seem that the introduction of the dilaton complicates the expression for the deviation from Kepler enormously. However, there is a massive simplification in the case where the charges are extremal, whose special nature was also highlighted in Section~\ref{EMD sec}. In the present case of a two-body system, we will have to take both charges extremal and of opposite sign, see~\eqref{extremal-charge} -- when taking the same extremal sign the static forces cancel out and the effective Newton's constant $G_{12}$ vanishes. Instead, when taking opposite signs, all forces are attractive and hence add up in the 0PN Hamiltonian. Furthermore, in the 1PN Hamiltonian, the parameters $\beta_{1,2}$ vanish entirely, leading to the simple result
\begin{equation} 
    \delta \phi_{\rm 1PN, EMD}|_{\rm ext.}  = \pi\frac{4(1+a^2)M^2\mu^2}{L^2}(3-a^2)\,.
\end{equation}
We therefore find that at $a^2=3$, this relativistic system of extremal black holes becomes equivalent to Kepler\footnote{This value coincides with the Kaluza-Klein reduction of gravity in $5$ dimensions~\cite{Cornish}.}. It has a LRL vector and therefore $SO(4)$ hidden symmetry. Moreover, the orbit closes as the perihelion precession vanishes.

This result is closely related to the findings for extremal black holes in maximal supergravity ~\cite{Caron-Huot:2018ape}. The role of the $SU(8)$ charge vector misalignment in maximal supergravity, needed in order to create a nonzero force between the extremal objects other than velocity dependent forces, is played in our case by the opposite nature of the charges\footnote{One could further extend our considerations and include magnetic charges as well. We expect the dyonic charges to span a $U(1)$ charge vector playing a completely analogous role to the $SU(8)$ charge vector of~\cite{Caron-Huot:2018ape}.}. In contrast to the rigid nature of maximal supergravity, enforced by the $N=8$ supersymmetry, we have the freedom to tune the dilaton coupling, finding that the two-body systems of extremal and anti-extremal black holes always are a special case with a particularly simple expression for $\Delta$, but that this only corresponds to Kepler-dynamics for a particular dilaton coupling.


\subsection{Kepler dynamics in the test-mass limit} 

Above, we have shown that the dynamics of the first relativistic correction of a system with comparably-sized masses in EMD theory behaves just like the classical Kepler problem. Now, we wish to extend our analysis to higher orders and will consider another tractable limit: that of the test-mass limit ($m_1\ll m_2$). Again, we can show the equivalence of this system with opposite and extremal charges in EMD with $a=\sqrt{3}$ to a Kepler-like system. However, in this system we can include all relativistic corrections. 

We will focus immediately on the case with extremal charges. The general (scalar) charged black hole metric simplifies significantly in the extremal limit, and it will be convenient to use the Majumdar-Papapetrou solution in the isotropic coordinate system~\eqref{eq: CGmetric}~\cite{Cornish}
\begin{equation}
    \dif s^2 = - U^{-2/(1+a^2)} \dif t^2 +U^{2/(1+a^2)} (\dif r^2 + r^2 \dif \theta^2)\,,
\end{equation}
with a single center: 
\begin{equation}
    U(r, \theta) = 1+(1+a^2)\frac{m_2}{r}\,.
\end{equation}
In the extremal case, the scalar field and vector are given by\footnote{In terms of the Schwarzschild radial coordinate, this choice of gauge corresponds to $A_0=-\frac{1}{\sqrt{1+a^2}}+\frac{m_2Q_2}{r}$. After the change $r\to r+r_\pm$, we find the above.}
\begin{equation}
    e^{a\phi}=U^{-a^2/(1+a^2)} \,, \qquad Q_1A_0 = U^{-1} \,,
\end{equation}
where we have chosen static gauge for the latter.

The Lagrangian for a point particle with charge $Q_1$ reads 
\begin{equation} \label{eq: EMDpp}
    L_{pp}=m_1e^{a\phi}\sqrt{-\dot{x}_\mu\dot{x}^\mu}+m_1Q_1A_\mu\dot{x}^\mu\,.
\end{equation}
The above is an extended Lagrangian, where the time coordinate can be seen as another dimension in the space; the Lagrangian is defined on the tangent space $T\Bar{M}$ of a $d+1$ dimensional manifold $\Bar{M}=\mathbb{R}\cross M$, the extended configuration manifold. All coordinates and velocities (denoted as a dot) are parametrised by a time-like variable $s$. Writing the Lagrangian in terms of the harmonic function we have
\begin{align}
    L_{pp}&=m_1U^{-1}\left( \sqrt{1-U^{4/(1+a^2)} \abs{\der[q]{t}}^2} +1 \right)\dot{t}\,.
\end{align}
Note that solutions to the Euler-Lagrange equations following from this action will not be unique, as different choices of time parametrisation will correspond to the same physical solution. This fact allows one to reduce the Hamiltonian of the system to an autonomous Hamiltonian on $T^*M$ instead, which is different from the one related by Legendre transform to the Lagrangian above~\cite{Marsden_2001}, being it defined on $T^*\Bar{M}$. This, in turns, leads directly to the relation to the classical Kepler problem. Since the theory was written down parametrization invariant, we can choose any monotonic function as time parameter~\cite{TongGR}. We will take the simple time parametrisation $\dot{t}=-1$.
The Legendre transform then results in
\begin{equation}
    H(q,\der[q]{t})
        =\frac{\partial L_{pp}}{\partial \der[q]{t}}\cdot \der[q]{t} - L_{pp}
        =m_1U^{-1}\left( \frac{1}{\sqrt{1-U^{4/(1+a^2)} \abs{\der[q]{t}}^2}} +1 \right) \,,
\end{equation}
where $t=x^0$ is the time of the particle seen from the rest frame of the central mass and $q=(x^1,x^2,x^3)$ the position. 
Solving for the momenta conjugate to the positions,
\begin{equation}
    p_i=m_1\frac{\partial L_{pp}}{\partial \der[q^i]{t}}=\frac{U^{(3-a^2)/(1+a^2)} \der[q_i]{t} }{\sqrt{1-U^{4/(1+a^2)} \abs{\der[q]{t}}^2}}\,, 
\end{equation}
then leads to the Hamiltonian in phase space
\begin{equation}
    H(q,p)=m_1 U^{-1} \left( \sqrt{ 1 + U^{2(a^2-1)/(1+a^2)} \frac{\abs{p}^2}{m_1^2}} +1 \right) \,.
\end{equation}
Note that the rest-mass energy is equal to $2m_1 c^2$; this differs from the usual $m_1 c^2$ due to the specific gauge choice that we have made for the gauge vector.

There is a number of interesting subcases to consider. First of all, 
the case $a=1$ leads to a Hamiltonian that is conformal to the special relativistic case,
 \begin{align}
     H(q,p)=m_1U^{-1}(q) \left( \sqrt{ 1 +  \frac{{p}^2}{m_1^2}} + 1 \right) \,.
 \end{align}
Instead, our main interest will be the case $a^2=3$ again. In this case we have
\begin{equation}
H(q,{p})=m_1U^{-1}(q) \left( \sqrt{ 1 + U(q) \frac{{p}^2}{m_1^2}} +1 \right) \,.
\end{equation}
Remarkably, this Hamiltonian satisfies the interesting relation 
\begin{equation} \label{eq: special EMD}
    \frac{1}{2}\left(\frac{H^2(q,p)}{m_1}-2  H(q,p)\right)=\frac{p^2}{2m_1}- \frac{2m_2 H^2(q,p)}{m_1r(q)}\,,
\end{equation}
where $r(q)=\abs{q}$. Shifting the Hamiltonian by the rest-mass energy and rescaling the distance by a factor $8$, one obtains (in terms of the new Hamiltonian)
\begin{equation} \label{eq: specEMD2}
    H(q,p)+\frac{1}{2}\frac{H^2(q,p)}{ m_1} =\frac{p^2}{2m_1}- m_2\frac{m_1+H(q,p)+ \frac{1}{4m_1}H^2(q,p)}{r(q)}\,.
\end{equation}
This specific form of the Hamiltonian shows that, following the arguments of Section~\ref{Kepler-type hams}, the extremal EMD 1-centre system with $a=\sqrt{3}$ is equivalent to the classical Kepler problem. 
It therefore also has a hidden LRL symmetry as well as closed orbits.

\bigskip

The same special behaviour can also be seen from the perspective of the equations of motion. Adopting the parametrisation $\dot{x}_\mu\dot{x}^\mu=-1$, there are two conserved quantities from the Lagrangian~\eqref{eq: EMDpp}
\begin{align}
    L=m_1U^{(2-a^2)/(1+a^2)}r^2\dot{\theta}\,, \qquad  E=m_1U^{(-2-a^2)/(1+a^2)}\dot{t}+m_1Q_1A_0\,,
\end{align}
as angular momentum and energy. Using again $\dot{x}^2=-1$ we can state 
\begin{equation} \label{eq: velocity norm}
     -  U^2\left(\frac{E}{m_1}-U^{-1}\right)^2 +U^{2/(1+a^2)}\dot{r}^2 + \frac{L^2}{m_1^2r^2}U^{(2a^2-2)/(1+a^2)}=-1\,.
\end{equation}
It is useful to now take the Binet variable $u\equiv\frac{1}{r}$, with $u'$ as its derivative with respect to $\theta$ so that 
\begin{equation}
    \dot{r}=-u'\frac{L}{m_1}U^{(a^2-2)/(1+a^2)}\,,
\end{equation}
and we find for the equation of motion
\begin{equation}\label{eq: cancel}
        (u')^2 + u^2-  U^{4/(1+a^2)}\frac{1}{L^2}(E^2-2Em_1U^{-1})= 0\,.
\end{equation}
The last term here in principle provides an infinite expansion in increasing orders of $u$ (and its accompanying powers of $\frac{1}{c^2}$). However, if we now choose $a^2=3$, the powers of the harmonic function simplify and (restoring the gravitational constant) we have
\begin{equation}
    (u')^2 + \left(u-2\frac{Gm_2E^2}{L^2}\right)^2=  \frac{\left(E^2-2m_1E\right)}{L^2}+\frac{4G^2m_2^2E^4}{L^4}\,.
\end{equation}
Compare this to the classical equation of motion (see e.g.~\cite{TongGR})
\begin{equation}
    (u')^2 + \left(u-\frac{Gm_2 m_1^2}{L^2}\right)^2=  \frac{2E_N m_1}{L^2}+\frac{G^2m_2^2m_1^4}{L^4}\,,
\end{equation}
where $E_N$ is the Newtonian energy. We see the only difference resides in the modification of the gravitational constant by a function $g(E)=2\frac{E^2}{m_1^2}$. Accordingly, the Hamiltonian giving the Kepler-like structure in~\eqref{eq: special EMD} here coincides exactly with the role of the Newtonian energy. The orbits will therefore be the same up to the above modification of the gravitational constant.


\section{Conclusion}

This paper studies relativistic systems of gravitating bodies, with dynamics equivalent to the classical Kepler problem. In particular, we have shown a class of seemingly relativistic Hamiltonians to have proportional flow to the Keper Hamiltonian on a levelset and we provided the accompanying Laplace-Runge-Lenz vector. Moreover, to fifth order in the PN expansion, we were able to construct the symplectic transformations and energy redefinitions needed to transform the Kepler Hamiltonian into such Kepler-type Hamiltonians explicitly, beyond the levelset equivalence. Additionally, a conjecture was put forth that all relativistic systems of a certain kind, i.e. Kepler at zeroth order and PN corrections of the form 
\begin{equation}
    c_{n,m,l}\frac{(p^2)^n(p_r^2)^l}{r^m}\,,
\end{equation}
that conserve a (relativistic version of a) Laplace-Runge-Lenz vector are canonically conjugate up to time reparametrization to the Kepler system. This conjecture was also shown to hold at least to fifth PN order.  

Remarkably, this type of Hamiltonians is not merely a mathematical possibility, but it is actually realised in a comparatively simple and interesting physical theory. The Einstein-Maxwell-dilaton theory, when considering two extremal black holes with opposite signs of the charges and dilaton coupling tuned to the Kaluza-Klein reduction value ($a=\sqrt{3}$), has Hamiltonians of exactly this form in both the test-mass limit and the $1$PN expansion of the two-body system. 

We therefore have established an interesting link between relativistic Hamiltonians, the ordinary Kepler problem and an explicit realisation. Several directions for further exploration present themselves. Firstly, exploring the conditions for local and global existence of the implicit, Kepler-type Hamiltonians and studying the geometry of the corresponding phase space would make for an intriguing investigation. 

Secondly, as the equivalence to Kepler for the discussed Hamiltonians is only shown on a levelset, the full phase space will in general look different from the Kepler phase space. Roughly put, the constant energy surfaces are `stacked' in a different way in the Kepler-type systems as compared to the original Kepler system. This raises the question whether one can always find a symplectic transformation from one to the other, as we have shown explicitly to a limited order. While we expect the normal-form-like construction of canonical transformations to extend to higher orders, perhaps even arbitrarily high orders, there is no guarantee this procedure will converge. However, it would be very appealing, if possible, to construct the asymptotic series of the transformations. 
Extending our local relations (on or in a neighbourhood of an energy surface) to Kepler to global relations, would also address the question whether the so(4) algebra is indicative of a $SO(4)$ symmetry group.

Thirdly, in the non-relativistic Kepler problem, the geometrical origin of the $SO(4)$ symmetry of $3$-dimensional Kepler is known to stem from a mapping to the motion of a free particle on a three-sphere, as derived by Fock~\cite{Fock1935} in $1935$. In the context of the EMD system, we have a natural way of perturbing the Kepler problem, by allowing for example the dilaton coupling to deviate from $a=\sqrt{3}$. This allows one to investigate which elements of this geometric construction would survive such a perturbation in the mapping to the three-sphere. Can the motion still be described by free motion on some hypersurface?

Also related to the larger-than-expected symmetry group of the EMD $1$-centre system is the Kaluza-Klein reduction of $5$-dimensional Einstein-Hilbert gravity, yielding EMD with the special dilaton coupling. Can we understand the origin of the hidden symmetry from the higher-dimensional origin of its theory? After all, while an $SO(4)$ symmetry in $3$ dimensions might surprise the reader unfamiliar with the Kepler problem, this is simply the group of spatial rotations in 5D. It would be interesting to investigate this correspondence and possible relation further.

Closely connected to the latter point is the more involved theory of $\mathcal{N}=8$ supergravity, which can be obtained as the dimensional reduction of supergravity from 11 to 4 dimensions; many of our EMD findings were already highlighted in this setting from the perspective of vanishing periastrion precession~\cite{Caron-Huot:2018ape}. Moreover, extremal black holes in the $\mathcal{N}=8$ theory have vanishing periastron precession to third post-Minkowskian order~\cite{Parra-Martinez:2020dzs}, at least leaving open the possibility of conserving a LRL vector to higher order and relating to Kepler. It is not clear that this also applies to the higher order two-body Hamiltonians of the extremal EMD with $a=\sqrt{3}$; we leave this interesting question open for future study.

\section*{Acknowledgments}
We are grateful to Andreas Knauf, Tom\'{a}s Ort\'{i}n and C\'{e}dric Deffayet for stimulating discussions and to the anonymous referees for their useful comments. D.N. is supported by the Fundamentals of the Universe research program within the University of Groningen. M.S. is supported by the NWO project 613.009.10.

\printbibliography

\end{document}